\documentclass[11pt]{article}
\usepackage{geometry}
\usepackage[utf8x]{inputenc} 

\usepackage{epsfig}
\usepackage{amsmath,amssymb}
\def\Box{\rule{2mm}{2mm}}

\newtheorem{fact}{Fact}[section]

\newtheorem{definition}[fact]{Definition}
\newtheorem{claim}{Claim}
\newtheorem{lemma}{Lemma}
\newtheorem{theorem}{Theorem}

\newcommand{\junk}[1]{}

\newenvironment{proof}{\noindent {\it Proof.}}{\Box \vskip \belowdisplayskip}

\begin{document}

\title{Multicommodity Flows in Planar Graphs with Demands on Faces}
\author{Nikhil Kumar \\IIT Delhi}
\date{}
\maketitle

\begin{abstract}
We consider the problem of multicommodity flows in planar graphs. Seymour \cite{seymour1981odd} showed that if the union of supply and demand graphs is planar, then the cut condition is also sufficient for routing demands. Okamura-Seymour \cite{okamura1981multicommodity} showed that if all demands are incident on one face, then again cut condition is sufficient for routing demands. We consider a common generalization of these settings where the end points of each demand are on the same face of the planar graph. We show that if the source sink pairs on each face of the graph are such that sources and sinks appear contiguously on the cycle bounding the face, then the flow cut gap is at most 3. We come up with a notion of approximating demands on a face by convex combination of laminar demands to prove this result.
\end{abstract}

\section{Introduction}
Given an undirected graph $G$ with edge capacities and multiple source/sink pairs, each with an associated demand, the multicommodity flow problem is to route all demands simultaneously without violating edge capacities. The problem was first formulated in the context of VLSI routing in the 70s and since then it has seen a long and impressive line of work. 

The demand graph, $H$ is the graph obtained by including an edge $(s_{i},t_{i})$ for a demand with source/sink $s_{i},t_{i}$. A necessary condition for the flow to be routed is that demand across any cut does not exceed capacity. This condition is known as the $\textbf{cut condition}$ and is known to be sufficient when $G$ is planar and all the source/sink pairs are on one face \cite{okamura1981multicommodity} or when $G+H$ is planar \cite{seymour1981odd}. However, one can construct small instances where the cut condition is not sufficient for routing flow (see figure \ref{gap}). When $G$ is series-parallel, if every cut has capacity at least twice the demand across it, then flow is routable \cite{chakrabarti2008embeddings,gupta2004cuts}. The flow-cut gap of a certain graph class is the smallest $\alpha$ such that flow is routable when capacity of every cut is at least $\alpha$ times the demand across it. Thus, for series-parallel graphs, the flow-cut gap is 2. For general graphs, the flow-cut gap is $\theta(\log k)$ \cite{Linial1995}, where $k$ is the number of demand pairs. 

The flow-cut gap for planar graphs ($G$ planar, $H$ arbitrary) is $O(\sqrt{\log n})$ \cite{rao1999small} and is conjectured to be $O(1)$ \cite{gupta2004cuts}. Chekuri et al. \cite{chekuri2006embedding} showed a flow cut gap of $2^{O(k)}$ for $k$-outerplanar graphs. \cite{krauthgamer2019flow} made progress towards this conjecture by showing an $O(\log h)$ bound on the flow-cut gap, where $h$ is the number of faces having source/sink vertices. This was subsequently improved to $O(\sqrt{\log h})$ by Filster \cite{filtser2020face}. 

Seymour \cite{seymour1981odd} showed that if the union of supply and demand graphs is planar, then the cut condition is also sufficient for routing demands. Okamura-Seymour \cite{okamura1981multicommodity} showed that if all demands are incident on one face, then again cut condition is sufficient for routing demands. In this paper, we consider instances where the source and sink of every demand lie on the same face (called face instances) and show flow cut gap results for them. Note that this is a common generalization of the settings considered in \cite{okamura1981multicommodity} and \cite{seymour1981odd}. Also, the cut condition is not sufficient for such instances (figure \ref{gap}).
\begin{figure}[ht]
\centering
\includegraphics[width=2in]{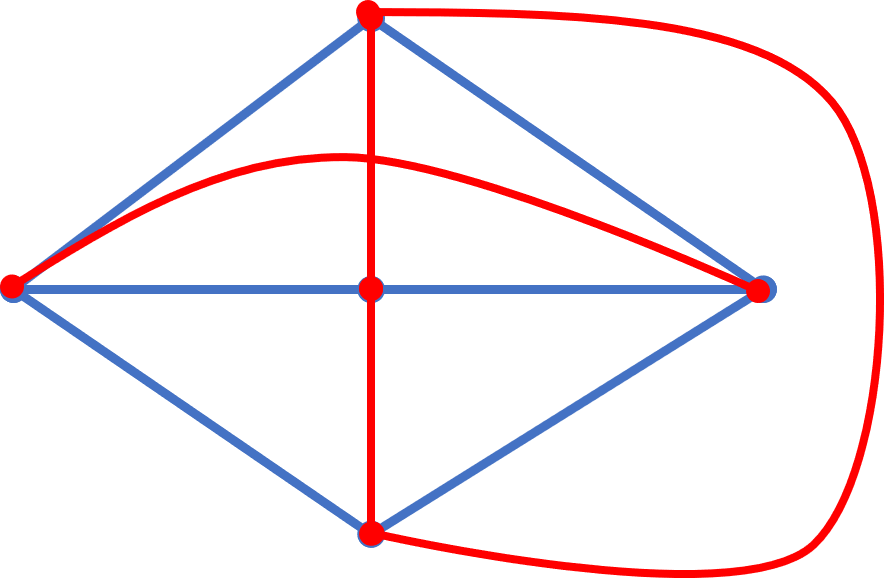}
\caption{Gap Instance: cut condition is satisfied but no feasible flow.  All supply (blue) and demand (red) edges have value 1. Since the end points of every demand edge is 2 units apart, a total of $4 \times 2=8$ supply edges are required for a feasible routing but only 6 are available.}
\label{gap}
\end{figure}

A common approach to establish bounds on the flow-cut gap is to bound the $L_{1}$ distortion incurred in embedding an arbitrary metric on the graph $G$ into a normed space. This, for instance, has been the method used to establish flow-cut gaps for general graphs \cite{Linial1995}, series parallel graphs \cite{chakrabarti2008embeddings,gupta2004cuts} and planar graphs \cite{rao1999small}. Our approach is very different. 
The central idea of our approach is to approximate an arbitrary set of demands on a face by a non-crossing family of demands. The flow cut gap then depends on the quality of this approximation. We formalize this notion in Section \ref{ch5:demand domination}. An instance is called separable face instance if the source sink pairs on each face of the graph are such that sources and sinks appear contiguously on the cycle bounding the face. Our main result is an upper bound of 3 for flow cut gap on separable face instances using the notion of approximation mentioned above. To the best of our knowledge, this is the first constant flow cut gap result for such instances. We extend our techniques to the setting when sources and sinks on a face may not be contiguous (called face instances); doing this incurs an additional $O(\log t)$ factor for us and matches the best known upper bound on flow cut gap for face instances by Naves et al. \cite{naves2010congestion}. \cite{naves2010congestion} concatenate the flow paths obtained after applying repeatedly a theorem of \cite{seymour1981odd}, while we crucially use planarity to get better results for separable face instances. Their approach can't be used to prove a constant flow cut gap for separable face instances. \cite{naves2010congestion} also show that their approach cannot be used to prove a constant flow cut gap for face instances while no such lower bound is known for our approach and it is possible that our approach may lead to a constant flow cut gap for face instances. Our proof also yields approximation algorithms to compute the generalized sparsest cut for such graphs and all our algorithms are combinatorial.

\section{Definitions and Preliminaries}

Let $G=(V,E)$ be a undirected graph with edge capacities $c:E \rightarrow{} \mathbb{Z}_{\geq 0}$. We call this the supply graph. Let $H=(V,D)$ be a graph with demands on edges $d:D \rightarrow{} \mathbb{Z}_{\geq 0}$. We call this the demand graph. We create $d(e)$ parallel copies of a demand edge $e \in F$ and assume that all demand edges have demand exactly 1. This assumption simplifies the presentation but it blows up the input size. All our algorithms can be modified to run in time polynomial in input size (see Section \ref{polytime}). The objective of the $\bf{multicommodity}$ $\bf{flow}$ problem is to find an assignment of positive real numbers to paths between the end points of demand edges in the supply graph such that the following hold: 
for every demand edge $(u,v) \in D$, total value of paths between $u$ and $v$ in $G$ is $d(e)$ and for every supply edge $e\in E$, total value of paths using it is at most $c(e)$. 

We say that an instance is feasible if paths satisfying the above two conditions can be found. We call an instance integrally feasible if there exists an assignment of non-negative integers to paths such that total number of paths in $G$ for every demand edge $e \in D$ is $d(e)$ and total number of paths using a supply edge $e \in E$ is at most $c(e)$. If the capacity of every edge in $G$ is 1, then the problem of finding integral flow is equivalent to finding edge disjoint paths (EDP) between the terminals. In this paper we will be concerned with fractional flows only. A cut $S \subseteq V$ is a partition of the vertex set $(S,V \setminus S)$. The sum of capacity of edges of $G$ going across the cut $S$ is denoted by $\delta_{G}(S)$. We abuse notation and use $\delta_{G}(S)$ to denote the set of edges crossing the cut $S$ as well. The meaning of notation used will be clear from the context. Similarly, $\delta_{H}(S)$ denotes the total value of demand edges going from $S$ to $V \setminus S$. One necessary condition for routing the flow is as follows: for every $S \subseteq V, \delta_{G}(S) \geq \delta_{H}(S)$. In other words, across every cut, total supply should be at least the total demand. This condition is also known as the cut condition. In general, cut condition is not sufficient for a feasible routing. We can ask for the following relaxation: given an instance for which the cut condition is satisfied, what is the maximum value of $f$, such that $f$ fraction of every demand can be routed? The number $f^{-1}$ is known as the $\textbf{flow-cut gap}$ of the instance. We will use an equivalent definition of the flow-cut gap: given an instance $(G,H)$ satisfying the cut condition, the smallest number $k \geq 1$, such that $(kG,H)$ is feasible, where $kG$ denotes the graph with every edge capacity multiplied by $k$. 

The following two classic results identify settings where the cut condition is also sufficient for routing demands in planar graphs. We will be invoking these to prove our results. 
\begin{theorem} [\cite{seymour1981odd}] \label{seymour}
If $G+H$ is planar, then cut-condition is necessary and sufficient for (half-integral) routing of all demands. Also, such a routing can be found in polynomial time.
\end{theorem}

\begin{theorem} [\cite{okamura1981multicommodity}]\label{OS}
If $G$ is a planar graph and all the edges of $H$ are restricted to a face, then cut condition is necessary and sufficient for (half-integral) routing of all demands. Also, such a routing can be found in polynomial time.
\end{theorem}

Instances considered in this paper (defined below) generalize the setting of above results.

\begin{definition}
$\textbf{Face Instance}$: $(G,H)$ is a face instance of multicommodity flow problem if $G$ is planar and for every demand edge $uv$, there exists a face $F$ such that $\{u,v \} \in F$. 
\end{definition}


In all our results, we assume a fixed planar embedding of the supply graph $G$. Without loss of generality, one can assume that $G$ is 2-vertex connected. If there is a cut vertex $v$ and $ab$ is a demand separated by removal of $v$, then replacing $ab$ by $av,vb$ maintains the cut condition. By doing this for every cut vertex and demand separated by them, we get separate smaller instances for each 2-vertex connected component. Note that a face instance remains one after the above operation. Hence, every vertex is a part of cycle corresponding to some face. By our assumption, for every demand edge there exists a face such that both its end points lie on that face. Hence, we can associate every demand with a face. Let $H_{F}$ denote the set of demands associated with a face $F$. We abuse notation and use $F$ to also denote the cycle associated with the face. Let $S \subseteq F$ be a contiguous segment of $F$. $\delta_{H_{F}}(S)$ will denote the total demand going from $S$ to $F /S$ in $H_{F}$, ie. $\delta_{H_{F}}(S)=|\{uv| uv \in H_{F},|\{u,v\} \cap S|=1\}|$. We say that demands on a face $F$ are $\bf{separable}$ if there exists a contiguous segment $S \subseteq F$ such that for any demand edge $uv \in H_{F},|S \cap \{u,v \}| = 1$. An instance is called a $\textbf{separable face instance}$ if it is a face instance and demands associated with all the faces are separable.
A cut minimizing the ratio of supply and demand across it is called the $\bf{sparsest}$ cut, ie. $\displaystyle \min_{S\subseteq V}\delta_{G}(S)/\delta_{H}(S)$. We call a subset $A \subseteq V$ central if both $G[A]$ and $G[V-A]$ are connected. The following is well-known but we give a proof for completeness.
\begin{lemma}\label{central}
$(G,H)$ satisfies the cut condition if and only if all central sets satisfy the cut condition. 
\end{lemma}
\begin{proof}
Clearly, if $(G,H)$ satisfy the cut condition for all sets then it is satisfied for the central sets. Suppose the cut condition is satisfied for all central sets but there is some non-central set $S'$ such that $\delta_{G}(S') < \delta_{H}(S')$. Choose $S'$ with minimal $\delta_{G}(S')$ among all such sets. We obtain a contradiction as follows. Let $S_1, S_2,\ldots, S_k$ be the connected components in $G \setminus \delta_{G}(S')$; since $S'$ is not central, $k ≥ 3$. Moreover each $\delta_{G}(S_i)$ is completely contained in $\delta_{G}(S')$. If some $j,S_j=S'$, then $\sum_{i:i \neq j}\delta_{H}(S_i)= \delta_{H}(S') >  \delta_{G}(S') = \sum_{i:i \neq j} \delta_{G}(S_i)$. Hence, there exists a $S_j$, different from $S'$ which violates the cut condition, ie. $\delta_{H}(S_j) > \delta_{G}(S_j)$. Moreover, by minimality in the choice of $S'$, $S_j$ is central, contradicting the assumption. If $S' \neq S_j$ for any $j$, then  $\sum_{i}\delta_{H}(S_i)= 2 \delta_{H}(S') >  2 \delta_{G}(S') = \sum_{i} \delta_{G}(S_i)$. Hence, there exists a $S_j$, different from $S'$ which violates the cut condition, ie. $\delta_{H}(S_j) > \delta_{G}(S_j)$. Again, by minimality in the choice of $S'$, $S_j$ is central, contradicting the assumption.
\end{proof}

The set of all faces of $G$ will be denoted by $F(G)$. The $\bf{dual}$ of a planar graph $G^{*}=(V^{*},E^{*})$ is defined as follows: $V^{*}= F(G)$ and if $f_{i},f_{j} \in F(G)$ share an edge in $G$, then $(f_{i},f_{j}) \in E^{*} $. It is a well known fact that edges of a central cut in $G$ correspond to a  simple circuit in $G^{*}$ and vice versa.
\section{Dominating Demands by Laminar Families} \label{ch5:demand domination}
Let $F$ be a face of $G$. A demand instance $H_{F}$ on $F$ is said to $\bf{dominate}$ demand $H'_{F}$ on $F$ if for all $S$ which are a contiguous segment of $F$, $\delta_{H_{F}}(S) \geq \delta_{H'_{F}}(S)$. We will denote this by $H_{F} \geq H'_{F}$.

We say that a pair of demand edges $uv,xy \in H_{F}$ $\bf{cross}$ if the terminals $\{u,v,x,y\}$ appear in order $u,x,v,y$ on the cycle $F$. We say that the set of demands $H_{F}$ on face $F$ is $\bf{laminar}$ if no two demands in $H_{F}$ cross each other. The main idea is to approximate $H_{F}$ by laminar instances. This is made formal in definition \ref{ch5:dominating demands}. Given a set of demands $H_{1},H_{2}$ on a face, $H_{1}+H_{2}$ is defined as the disjoint union of demands in $H_{1}$ and $H_{2}$. Recall that $H_1,H_2$ and $H_1+H_2$ can contain parallel edges and all edges have demand exactly 1. Given a positive integer $\alpha$ and a graph $G$, $\alpha G$ denotes the graph with edge capacities of $G$ multiplied by $\alpha$.  

\begin{definition} \label{ch5:dominating demands}
Let $F$ be a face and $H_{F}$ be the demands associated with it. A set of demands $L_{1},L_{2},\ldots,L_{k}$ on $F$ is said to be an $(\alpha_{1},\alpha_{2},\ldots,\alpha_{k})$ -approximation to $H$ if the following is true:
\begin{enumerate}
    \item $L_{i} \leq \alpha_{i} H_{F}$ for $1 \leq i \leq k$.
    \item $L_{1},L_{2},\ldots,L_{k}$ are laminar.
    \item $L_{1}+L_{2}+\ldots+L_{k} \geq H_{F}$.
\end{enumerate}
\end{definition}

For a particular face instance, suppose the demands on any face can be $(\alpha_{1},\ldots,\alpha_{k})$ approximated by laminar demands. We show that the flow cut gap of such an instance is at most $\sum_{i=1}^{k}\alpha_{i}$.

\begin{lemma} \label{approx}
Let $(G,H)$ be a face instance of multicommodity flow. If for every face $F$, $H_{F}$ can be $(\alpha_{1},\alpha_{2},\ldots,\alpha_{k})$ -approximated by laminar instances, then the flow cut gap of $(G,H)$ is at most $\sum_{i=1}^{k}\alpha_{i}$.
\end{lemma}

\begin{proof}
Let us assume that the cut condition is satisfied for $(G,H)$. We will show how to route the flow in $(\sum_{i=1}^{k}\alpha_{i})G$. We will do this in two phases. In phase 1, we will construct $k$ instances such that union of demand and supply graph is planar and cut condition is satisfied. Using Theorem \ref{seymour}, we will be able to find feasible flow paths for each of the $k$ instances. In phase 2, we will use the paths constructed in phase 1 to find a feasible routing of all the demands.

$\textbf{Phase 1}$: Let $L_{i}^{F},1 \leq i \leq k$ denote the $i$th set of  laminar demands for face $F$. We construct $k$ demand multicommodity flow instances as follows: $(G_{i},H_{i}): G_{i}=\alpha_{i}G, H_{i}= \bigcup_{F \in F(G)} L_{i}^{F}$ for $1 \leq i \leq k$. Note that by construction $G_{i} \cup H_{i}$ is planar. If we can show that the cut condition is satisfied for $(G_{i},H_{i})$, then by Theorem \ref{seymour}, a feasible routing of demands will exist. Using Lemma \ref{central}, to show that the cut condition is satisfied for $(G_{i},H_{i})$, we need to check the cut condition for central cuts only. Recall that a cut $S$ is central if $G[S]$ and $G[V-S]$ are both connected. Recall that any central cut in a plane graph corresponds to a simple cycle in its dual. Since in any simple cycle, degree of a vertex is either two or zero and a vertex in the dual corresponds to a face in the original graph, we can conclude that any central cut contains either two or zero edges of a face $F$ (ie. cycle corresponding to face $F$). Let $S$ be a central cut and suppose it crosses face $F_{1},F_{2},\ldots,F_{l}$, ie. $|\delta_{G}(S) \cap F_{j}|=2$ for $1 \leq j \leq l$. $S$ splits each of face it crosses into two segments, say $S_{j},F_{j}/S_{j}$. We know that for any $i,j$, demand going from $S_{j}$ to $F_{j}/S_{j}$ in $L_{i}^{F_{j}}$ is at most $\alpha_{i} \delta_{H_{F_{j}}}(S_{j})$. Summing over all faces that $S$ crosses, we get $\delta_{H_{i}}(S)\leq \alpha_{i} \delta_{H}(S) \leq \alpha_{i} \delta_{G}(S)$ (the second inequality is true because cut condition is satisfied for $(G,H)$. Hence, the cut condition is satisfied for $(G_{i},H_{i}),1 \leq i \leq k$ and a feasible routing of demands exists.

$\textbf{Phase 2}$: Consider a fixed routing of $(G_{i},H_{i})$ in phase 1. Fix a face $F$. From phase 1, we have flow paths $P_{i}^{F}$ for each set of demands $L_{i}^{F},1 \leq i \leq k$. Let $P^{F}= \bigcup_{i=1}^{k}P_{i}^{F}$. Now consider the multicommodity flow instance $(G_{F},H_{F})$ defined as follows: the supply graph $G_{F}=P^{F}$ and $H_{F}$ is the demands associated with face $F$ in the original instance. Since the union of $L_{i}^{F}$ dominate $H_{F}$ (property 3 of definition \ref{ch5:dominating demands}), across every (central) cut, the number of supply edges (ie. paths used to route $L_{i}^{F}$) is more than the number of demand edges in $H_{F}$. Observe that this is just the setting of Theorem \ref{OS}, which states that if end points of all the demand lie on one face, then the cut condition is necessary and sufficient for routing. Hence, all demands associated with face $F$ can be routed using $P^{F}$. Doing this for all the faces gives the desired routing. 
\end{proof}

\section{Constructing Laminar Families}

In this section, we show how to construct a family of laminar demands which approximate original demands well. We will crucially use a notion of uncrossing for this construction. Consider a pair of crossing demands $uv,xy \in H_{F}$. We $\textbf{uncross}$ these demands by replacing them with demands $ux,vy$. The following lemma shows that uncrossing a pair of demands does not increase the total demand across any cut.

\begin{lemma}\label{uncross}
Let $F$ be a cycle with demand $H_{F}$ incident on it. Let $uv,xy \in H_{F}$ be a crossing pair and $H'_{F}$ be the set of demand created by replacing $uv,xy$ by $ux,vy$. Then for any contiguous segment $S \subseteq F, \delta_{H_{F}}(S) \geq \delta_{H'_{F}}(S)$.
\end{lemma}

\begin{proof}
Consider the segments $[ux],[xv],[vy],[yu]$. A cut is defined by two edges $e_{1},e_{2} \in F$. Let $e_{1} \in [uy],e_{2} \in [vx]$. Replacing $uv,xy$ by $ux,vy$ reduces the demand across this cut by 2. Let $e_{1} \in [uy],e_{2} \in [ux]$. Replacing $uv,xy$ by $ux,vy$ doesn't reduce the demand across this cut. For the remaining choices of $e_1,e_2$, it is easy to check that the demand across the cut remains the same and the lemma follows.
\end{proof}

Recall that demands on a face $F$ are separable if there exists a contiguous segment $S$ on $F$ such that all demands go from $S$ to $F\setminus S$. The following lemma shows how to construct a good approximating laminar family in case of separable demands. 

\begin{theorem}\label{clubbed}
Given any face $F$, demands $D$ on it and a contiguous set $S \subseteq F$ such that for any $uv \in D$, $|\{u,v\} \cap S|=1$, there exists instances of laminar demands $L_{1},L_{2}$ such that $L_{1} \leq D, L_{2} \leq 2D$ and $L_{1}+L_{2} \geq D$. Also, $L_{1}$ and $L_{2}$ can be found in polynomial time.
\end{theorem}

\begin{proof}
We create $L_{1},L_{2}$ in two phases, $\textbf{Phase 1}$ and $\textbf{Phase 2}$ respectively (see Figure \ref{fig2}). Recall that all demand edges have demand exactly 1. By making multiple copies of vertices, we also assume that at most one demand edge is incident on any vertex. Now we describe the two phases. Let the endpoints of demands belonging to $S$ be called $s_{1},s_{2},\ldots,s_{k}$ and the other end points be called $t_{i}'s$. After renaming, $\{ (s_{1},t_{1}),(s_{2},t_{2}),\ldots,(s_{k},t_{k}) \} $ are the demand edges in $D$ and $s_{1},s_{2},\ldots,s_{k}$ appear in that order on $F$. See Figure \ref{fig2}.

$\textbf{Phase 1}$: In this phase, we uncross the demands so as to maintain the following property: exactly one end point of a demand is in $S$. Given any two crossing demands $s_{i}t_{i}$ and $s_{j}t_{j}$, we replace them by uncrossed demands $s_{i}t_{j}$ and $s_{j}t_{i}$. We keep on repeating this process while there are crossing demands. Note that whenever we uncross a pair of crossing demands, total number of pair of crossing demands decrease by exactly 2. This implies that the uncrossing procedure stops after a finite number of steps. This forms the first laminar instance $L_{1}$. Let the demands in $L_{1}$ be $s_{i}t_{\sigma(i)}$ for $1 \leq i \leq k$ (see figure \ref{fig2}).
    
$\textbf{Phase 2}$: In this phase, we make sure while uncrossing that cuts crossing $S$ have sufficiently large value (by cuts crossing $S$ we mean cuts $S'$ such that $S \cap S'\neq \phi$). To do this, we take a demand $s_{i}t_{\sigma(j)}$ and replace it by $t_{\sigma(i)}t_{\sigma(j)}$. We do this for all $1 \leq i \leq k$. The new demand edges are of the form $t_{\sigma(i)}t_{\sigma(j)}$ and may be crossing. Note that exactly 2 demand edges are incident on $t_i$'s and no demand edge is incident on $s_i$'s. Remove self loops, if any. Call this demand instance $L_{2}'$. We further modify $L_{2}'$ by uncrossing demands as follows: let $t_{\sigma(i)}t_{\sigma(j)}$ and $t_{\sigma(k)}t_{\sigma(l)}$ be intersecting demands such that $i<k<j<l$. We replace such a demand pair by $t_{\sigma(i)}t_{\sigma(l)}$ and $t_{\sigma(k)}t_{\sigma(j)}$. We keep on repeating this procedure until no such crossing pair of demands remain. Let $L_2$ be the uncrossed instance formed at the end of this process (see figure \ref{fig2}).
    
\begin{figure}[ht] 
\centering
\includegraphics[width=5in]{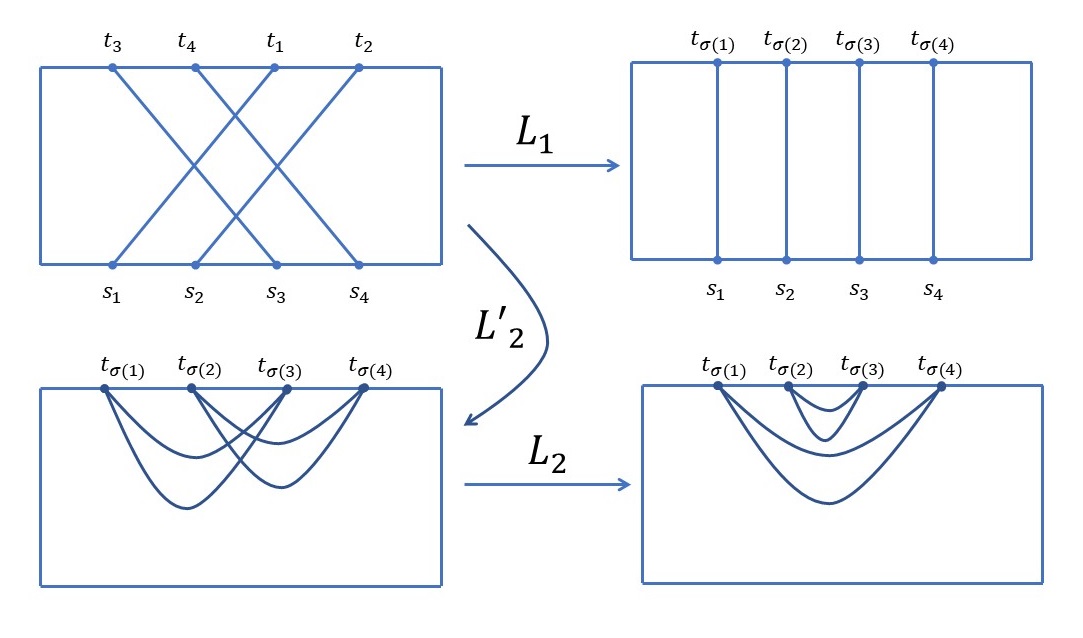}
\caption{Creating Laminar Demands}
\label{fig2}
\end{figure}

The following claims complete the proof of the lemma by showing that $L_1,L_2$ are $(1,2)$ dominating.
\begin{claim}
Laminar instances $L_{1},L_{2}$ formed by the procedure above satisfy $L_{1} \leq D$, $L_{2} \leq 2D$. 
\end{claim}
\begin{proof}
$L_{1}$ is formed from the original instance $D$ by uncrossing pairs of crossing demand edges repeatedly. By Lemma \ref{uncross}, we know that this doesn't increase the total demand across any cut and hence $L_{1} \leq D$. Consider the instance $L_{1}+D$. Since $L_{1} \leq D, L_{1}+D \leq 2D$. Now, we will show that $L_{2}' \leq L_{1}+D$. Demand edges in $L_{1} + D$ can be paired as follows: $(s_{i}t_{i},s_{i}t_{\sigma(i)})$ for all $1 \leq i \leq k$. Observe that replacing the $i$th edge pair by $t_{\sigma(i)}t_{i}$ does not increase the total demand across any cut. Note that $L_2$' is the set of demand edges formed by replacing $(s_{i}t_{i},s_{i}t_{\sigma(i)})$ by $t_{\sigma(i)}t_{i}$ for $1 \leq i \leq k$. Hence $L_{2}' \leq L_{1}+D$. Since, $L_{2}$ is formed by uncrossing pairs of demand in $L_{2}'$, by Lemma \ref{uncross}, $L_{2} \leq L_{2}' \leq L_{1}+D \leq 2D$ and the claim follows.
\end{proof}
\begin{figure}[t]
\centering
\includegraphics[width=4.5 in]{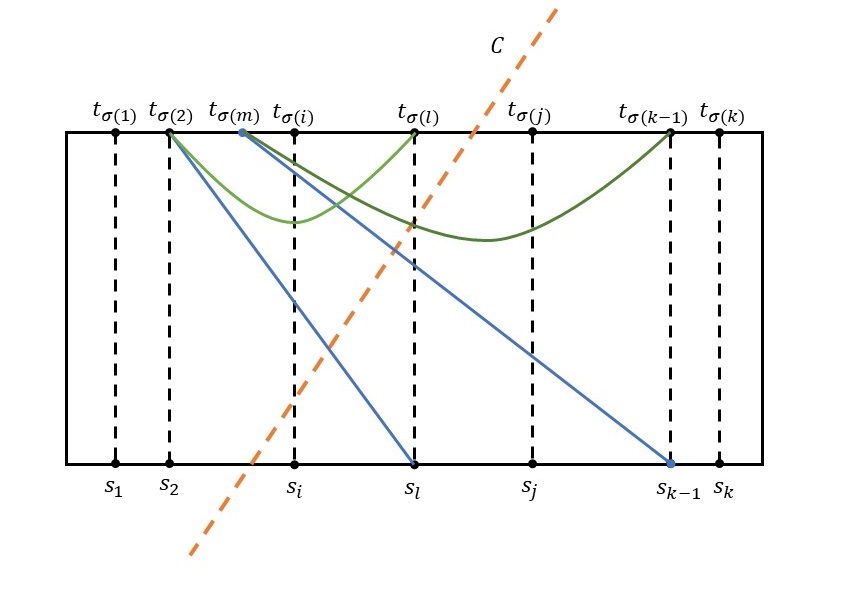}
\caption{$C$ is a cut of type 4. $|\{s_{l},t_{\sigma(2)}\} \cap C|=1$ but $|\{ t_{\sigma(2)},t_{\sigma(l)}\} \cap C| \neq 1$. $|\{s_{k-1},t_{\sigma(m)}\} \cap C|=1$ and $|\{ t_{\sigma(m)},t_{\sigma(k-1)}\} \cap C| = 1$.}
\label{fig3}
\end{figure}

\begin{lemma}
The union of $L_{1}$ and $L_{2}$ dominates the original demand $D$, ie. $L_{1}+L_{2} \geq D$.
\end{lemma}
\begin{proof}
$F$ is the face (cycle) on which demands are incident. Recall that we are only concerned with cuts that divide $F$ into two contiguous segments. Such cuts are described exactly by two edges $e_{1},e_{2} \in F$. Terminals $s_{1},s_{2},\ldots,s_{k},t_{1},t_{2},\ldots t_{k}$ divide $F$ into $2k$ segments. These segments define four kind of cuts:
\begin{enumerate}
    \item $e_{1} \in \{[s_{1}t_{\sigma(1)}],[s_{k}t_{\sigma(k)}]\},e_{2} \in F$. 
    \item $e_{1} \in [s_{i}s_{i+1}],e_{2} \in [s_{j}s_{j+1}]$.
    \item $e_{1} \in [t_{\sigma(i)}t_{\sigma(i+1)}],e_{2} \in [t_{\sigma(j)}t_{\sigma(j+1)}]$.
    \item $e_{1} \in [s_{i}s_{i+1}],e_{2} \in [t_{\sigma(j)}t_{\sigma(j+1)}]$.
\end{enumerate}
Observe that in cases 1-3, one side of the cut formed contains terminals of either only $s$ or only $t$ type. In original demand $D$, all the demands were of the form $s_{j}t_{j}$, hence the number of demand edges going across any cut of type 1-3 in $D$ is equal to the number of terminals contained in this cut. In $L_{1}$, we ensure that after uncrossing, all the demands are still of type $s_{j}t_{\sigma(j)}$ and hence the number of demand edges going across such cuts is equal in $D$ and $L_{1}$.

Let $C$ be a cut of type 4 (Figure \ref{fig3}). Recall that in phase 1, $\sigma$ was defined such that $s_{i}$ is paired with $t_{\sigma(i)}$ in $L_{1}$. There exists a traversal of $F$ such that $s_{1},s_{2},\ldots,s_{k},$ $t_{\sigma(k)},t_{\sigma(k-1)},\ldots,t_{\sigma(1)}$ appear in order. Hence, without loss of generality, cut in this case can be assumed to be $C= \{s_{i},s_{i+1},\ldots,s_{k},t_{\sigma(k)},t_{\sigma(k-1)},\ldots,t_{\sigma(j)} \}$. Let $i<j$. The other case is analogous. 

\begin{claim}
Let $s_{l}t_{\sigma(m)} \in D$ be a demand such that $|\{s_{l},t_{\sigma(m)}\} \cap C|=1$ but $|\{ t_{\sigma(l)},t_{\sigma(m)}\} \cap C| \neq 1$. Number of such demands is at most $|i-j|$.
\end{claim}
\begin{proof}
If $l < i$, then $s_{l}t_{\sigma(m)}$ crosses $C$ if $m>j$. In this case $t_{\sigma(m)}t_{\sigma(l)}$ also crosses $C$. If $l \geq j$, then $s_{l}t_{\sigma(m)}$ crosses $C$ if $m<j$. In this case $t_{\sigma(m)}t_{\sigma(l)}$ also crosses $C$. Hence, only demand edges for which $|\{s_{l},t_{\sigma(m)}\} \cap C |=1$ and $|\{ t_{\sigma(l)},t_{\sigma(m)}\} \cap C| \neq 1$ is true have to satisfy $i \leq l < j$ and there can be at most $|i-j|$ of them.
\end{proof}
Recall that $L_{2}'$ was formed by short cutting the demands of type $s_{i}t_{\sigma(j)},s_{i}t_{\sigma(i)}$ to $t_{\sigma(i)}t_{\sigma(j)}$. From the argument above, it follows that $\delta_{D}(C)-\delta_{L_{2}'}(C) \leq |i-j|$. Also, from phase 1 uncrossing, we have  $\delta_{L_{1}}(C) = |i-j|+1$ (see Figure \ref{fig3}). Hence, $\delta_{L_{2}'}(C) + \delta_{L_{1}}(C) \geq \delta_{D}(C)$.
  
Recall that in phase 2,  $L_{2}$ is created by uncrossing demands in $L_{2}'$ as follows: if there is a demand $t_{\sigma(i)}t_{\sigma(j)},t_{\sigma(k)}t_{\sigma(l)}$ with $i<k<j<l$, then replace it by $t_{\sigma(i)}t_{\sigma(l)},t_{\sigma(k)}t_{\sigma(j)}$. Using an argument similar to Lemma \ref{uncross}, it can easily be verified that such uncrossing preserves the number of demand edges going across $C$, which implies $\delta_{L_{2}}(C)=\delta_{L_{2}'}(C)$. Therefore, $\delta_{L_{2}}(C) + \delta_{L_{1}}(C) \geq \delta_{D}(C)$ and the claim follows.
\end{proof}

This completes the proof that $L_{1},L_{2}$ are $(1,2)$ approximate laminar family for $D$. \end{proof}

\begin{theorem}\label{clubbed theorem}
Let $(G,H)$ be a separable face instance of multicommodity flow . Then, the flow-cut gap of $(G,H)$ is at most 3.
\end{theorem}
\begin{proof}
Follows from Lemma \ref{approx} and Theorem \ref{clubbed}.
\end{proof}

\subsection{Dominating Laminar Families for Arbitrary Demands}
\begin{figure}[htb]
\includegraphics[width=5.5in]{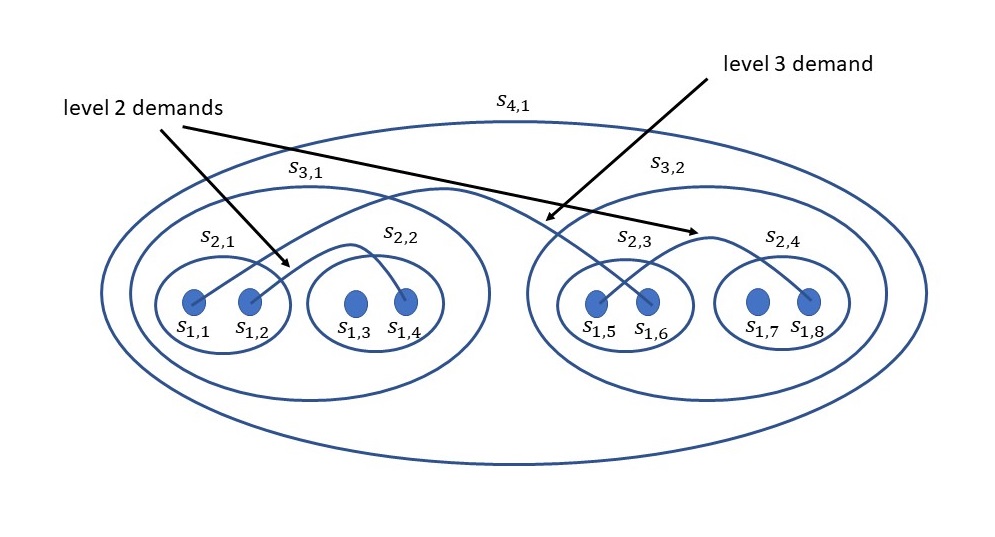}
\caption{Laminar set system $S$}
\label{fig4}
\end{figure}

\begin{lemma}\label{set system}
Given any face with demands $D$ and $t$ source/sink terminals incident on it, there exists instances of laminar demands $D_{1},D_{2},\ldots,D_{k}$ such that each of $D_{i}$ is dominated by $2D$, the union of $D_{i}$ dominate $D$ and $k=O(\log_{2}t)$.
\end{lemma}
\begin{proof}
Number the source/sink terminals from $1$ to $t$, starting from an arbitrary vertex and moving in clockwise direction. Assume that $t$ is a power of 2, otherwise add some dummy source/sink terminals. Define a laminar family of sets $S_{i,j}$ as follows: $S_{1,j}:=\{j\}$ for $1 \leq j \leq t$ and  $S_{i+1,j}=S_{i,2j} \cup S_{i,2j-1}$ for $1 \leq j \leq t/2^{i-1}, 1 \leq i \leq \log_{2}t+1$. The sets $S_{i,k}, 1\leq k \leq t/2^{i-1}$ are said to belong to level $i$ (figure \ref{fig4}). Let $uv$ be a demand with $u$ having a  numbering smaller than $v$ in the ordering. The demand $uv$ is said to belong to level $i$ if $u \in S_{i,2j-1}$ and $v \in S_{i,2j}$ for some $j$. Consider the lowest level set that contains both end points of a demand, say $S_{i,j}$ (such a set always exists). Then the demand must belong to level $i-1$ as its endpoints must be in $S_{i-1,2j-1}$ and $S_{i-1,2j}$. Also, every demand belongs to a unique level. Consider the demands in level $i$, say $D_{i}$. It is a disjoint union of separable instances, say $D_{i}^{m},1 \leq m \leq k$. For each of these disjoint instances, we can get 2 laminar instances with property guaranteed by Theorem \ref{clubbed}, say $L_{i}^{m}(1),L_{i}^m(2), 1 \leq m \leq k$. Let $L_{i}(1)= \bigcup_{i=1}^{k}L_{i}^{m}(1),L_{i}(2)= \bigcup_{i=1}^{t}L_{i}^{m}(1)$. Then, $L_{i}(1) \leq D_{i},L_{i}(2) \leq 2D_{i}, L_{i}(1)+L_{i}(2) \geq D_{i}$. Since, $D=\sum_{i=1}^{\log_{2}t}D_{i}$, we have $\sum_{i=1}^{\log_{2}t} L_{i}(1) \leq D$,$\sum_{i=1}^{\log_{2}t} L_{i}(2) \leq 2D$ and $\sum_{i=1}^{\log_{2}t} L_{i}(1)+L_{i}(2) \geq D$. 
\end{proof}

\begin{theorem} \label{general theorem}
Let $(G,H)$ be a face instance of multicommodity flow. Then, the flow-cut gap of $(G,H)$ is $O(\log t)$, where $t$ is the maximum number of terminals on any face.
\end{theorem}

\begin{proof}
Follows from Lemma \ref{approx} and Lemma \ref{set system}.
\end{proof}

\section{Sparsest Cut}
Let $(G,H)$ be an instance such that $G+H$ is planar. Then, it is easy to check if there exists a cut $S$ such that $\delta_{H}(S) > \delta_{G}(S)$ \cite{seymour1981odd}. Consider the planar dual $(G+H)^{*}$ of $G+H$. For every edge $e \in (G+H)^{*}$, define $w(e)$ as follows: If $e$ corresponds to a supply edge in $G$, then $w(e)=c(e)$, else if $e$ corresponds to a demand edge in $G$, then $w(e)=-d(e)$. Since every cut in $G$ corresponds to a cycle in $(G+H)^{*}$, if there is a cut $S$ such that $\delta_{H}(S) > \delta_{G}(S)$, then the corresponding cycle in the dual will have negative weight and vice versa. Hence, finding a violated cut is equivalent to finding a cycle of negative weight in a graph, which can be done efficiently \cite{bellman1958routing}. In the following, we show how to find $O(\log t)$ sparsest cut for face instances. Same argument also gives a 3-approximate sparsest cut for separable face instances. 
\begin{lemma}\label{sparsest cut}
Let $(G,H)$ be an instance of face multicommodity flow. Then, in polynomial time one can either find a cut $S$ such that $\delta_{H}(S) > \delta_{G}(S)$ or route all the demands in $O(\log t)G$, where $t$ is the maximum number of terminals on any face.
\end{lemma}
\begin{proof}
As stated in Lemma \ref{approx} and Lemma \ref{set system}, we create instances $(G_{i},H_{i}), 1 \leq i \leq k$, where $k=O(\log t)$. If all the instances are feasible, then using Lemma \ref{approx}, we can route all the demands in $kG$. If any of the instances $G_{i}$ is infeasible, we get a violated cut in the corresponding instance. From the above discussion, all our procedures can be made to run in polynomial time.
\end{proof}

\begin{theorem}
Let $(G,H)$ be an instance of face multicommodity flow. Then, in polynomial (in input size) time one can find a $O(\log t)$-approximate sparsest cut, where $t$ is the maximum number of terminals on any face.
\end{theorem}
 \begin{proof}
 Find the largest value $\lambda$ (by using binary search), such that $(G_{i},\lambda H_{i}),1 \leq i \leq k$ are feasible. By Lemma \ref{approx}, $(kG, \lambda H)$ is feasible. Since, $\lambda$ is the largest such value, there exists a cut $S$ such that $\lambda \delta_{H}(S) \geq \delta_{G}(S)$. Let $S^{*}$ be the optimal sparsest cut. Since $(kG, \lambda H)$ is feasible, $k \delta_{G}(S^{*}) \geq \lambda \delta_{H}(S^{*})$. Hence, $k \delta_{G}(S^{*})/ \delta_{H}(S^{*}) \geq \lambda \geq \delta_{G}(S) / \delta_{H}(S)$. Hence, $S$ is a $k=O(\log t)$-approximate sparsest cut.
 \end{proof}
 
\section{Uncrossing Demands in Polynomial Time} \label{polytime}
We give an implementation of the uncrossing in polynomial time. Given a face $F$ with $t$ vertices and a set of arbitrary demands incident on it, we wish to bound the number of iterations required to uncross the demands. Let $u_1,u_2,\ldots,u_t$ be the vertices on the cycle of face $F$ in that order. Let $d_e$ be the demand between vertices $u_i$ and $u_j$, where $e=(u_i,u_j)$. If there is no demand edge between $u_i,u_j$ in the original instance, we introduce a demand edge $(u_i,u_j)$ with $d_{(u_i,u_j)}=0$ and assume that there is a demand edge between every pair of vertices. Number of demand edges is at most $t^{2}$. Let $d_{\max}$ be the maximum demand. Let $D=\sum_{i,j : e_i,e_j \tt{ cross}} d_{e_{i}} d_{e_{j}}$.

We now describe the uncrossing procedure. Let $e_i,e_j$ be such that $d_{e_{i}} d_{e_{j}} \geq D/t^{2}$. Note that such a demand always exists as $D$ is a sum of at most $t^{2}$ such terms. Suppose $d_{e_{i}} \geq d_{e_{j}}$. There are two possible ways to uncross a pair of crossing demands, we describe both of them. We replace $e_i=(u_a,u_c),e_j=(u_b,u_d),a < b < c < d$ by 3 edges:\\
$\textbf{Possibility 1}$: $(u_a,u_b)$ with demand value $d_{e_{j}}$, $(u_c,u_d)$ with demand value $d_{e_{j}}$ and $(u_a,u_c)$ with demand value $d_{e_{i}}-d_{e_{j}}$.\\
$\textbf{Possibility 2}$: $(u_b,u_c)$ with demand value $d_{e_{j}}$, $(u_a,u_d)$ with demand value $d_{e_{j}}$ and $(u_a,u_c)$ with demand value $d_{e_{i}} - d_{e_{j}}$.

We replace any parallel edges created due to above procedure by a single edge with demand as the sum of demand values of all the parallel edges. We repeat the above uncrossing procedure until no crossing pair remain. It is easy to verify that after one iteration of uncrossing, the value of $D$ decreases by at least $d_{e_{i}} d_{e_{j}}$. Since, $d_{e_{i}} d_{e_{j}} \geq D/t^{2}$, $D$ decreases by a multiplicative factor of at least  $(1-1/t^{2})$ after every iteration. After $t^{2} \ln D$ iterations, its value is at most $D(1-1/t^{2})^{t^{2} \ln D} < D e ^{- \ln D} = 1$ and the uncrossing procedure terminates after at most $t^{2} \ln D=O(t^{2}\ln t + t^{2}\ln d_{\max})$ iterations.

\section{Conclusions and Open Problems} \label{ch5:conclusion}
We showed how to construct $(\alpha_{1},\alpha_{2},\ldots,\alpha_{k})$ laminar dominating family such that $\sum_{i=1}^{k}\alpha_{i} =O(\log t)$. We believe that our methods can be extended to get a family with $\sum_{i=1}^{k}\alpha_{i}=O(1)$. The best known lower bound on the flow-cut gap for face instances is 4/3 (Figure \ref{gap}) and it is an interesting open question to improve it. There is a tight relationship between the flow-cut gap and $L_{1}$ embedding of the shortest path metric of supply graph into normed space. Is it possible to prove (and improve) our results by metric embedding techniques?

$\textbf{Acknowledgement}$: I would like to thank Naveen Garg for useful discussions.
\bibliography{WGMV.bib}
\end{document}